\documentclass[12pt,twoside,reqno]{amsart}
\usepackage[colorlinks=false,citecolor=blue]{hyperref}
\usepackage{mathptmx, amsmath, amssymb, amsfonts, amsthm, mathptmx, enumerate, color,mathrsfs}
\usepackage{graphicx}

\usepackage{epstopdf}
\newtheorem{theorem}{Theorem}[section]
\newtheorem{lemma}{Lemma}[section]
\newtheorem{remark}{Remark}[section]
\newtheorem{definition}{Definition}[section]
\newtheorem{corollary}{Corollary}[section]

\newtheorem{proposition}{Proposition}[section]
\numberwithin{equation}{section}

\newcommand{\tr}{{\rm Tr\hskip -0.2em}~}
  \newcommand{\df}[2]{\frac{d#1}{d#2}}

\newcommand{\vertiii}[1]{{\left\vert\kern-0.25ex\left\vert\kern-0.25ex\left\vert #1 
    \right\vert\kern-0.25ex\right\vert\kern-0.25ex\right\vert}}

\begin{document}
\title{Bounds for the reduced relative entropies}
\author{Shigeru Furuichi and Frank Hansen}
\subjclass[2020]{47A63, 47A64, 94A17}
\keywords{Reduced relative entropy, reduced Tsallis relative entropy, variational expression, trace inequality, Golden--Thonpsom inequality}

\begin{abstract}
A lower bound of the reduced relative entropy is given by the use of
a variational expression. 
The reduced Tsallis relative entropy is defined and some results are given. 
In particular, the convexity of the reduced Tsallis relative entropy is obtained. 
Finally, an upper bound of the reduced Tsallis relative entropy is given. 

\end{abstract}

\maketitle

\section{Introduction}\label{sec1}
The second-named author  \cite{H2022} introduced the notion of reduced relative  entropy 
\begin{equation}\label{definition of reduced entropy}
S_H(A\mid B)=\tr[A\log A-H^*AH\log B-A+B],
\end{equation}
where $ H  $ is a contraction, and $ A $ and $ B $ are positive definite matrices. By an extension of Uhlmann's proof of convexity of the relative  entropy it was obtained, that also the reduced relative  entropy is convex. This also follows from the identity
\begin{align}
&S_H(\rho\mid\sigma)=\tr[\rho\log\rho-H^*\rho H\log\sigma-\rho+\sigma]\nonumber\\
&=\tr[HH^*\rho\log\rho-H^*\rho H\log\sigma]+\tr [(I-HH^*)\rho\log\rho]+\tr[-\rho+\sigma]\nonumber\\
&=S_f^H(\rho\parallel\sigma)+\tr [(I-HH^*)\rho\log\rho]+\tr[-\rho+\sigma],\label{quasi_form}
\end{align}
where the quasi-entropy \cite{P1986} 
\[
S_f^X(\rho\parallel\sigma)=\tr\, X^* \mathcal P_f(L_\rho,R_\sigma)X
=\tr[XX^*\rho\log\rho-X^*\rho X\log\sigma],
\]
is defined for an arbitrary matrix $X$.
The above \eqref{quasi_form} is calculated from $ f(t)=t\log t$ and $X:=H$. 
The joint convexity (resp. joint concavity) of $(\rho,\sigma)\to S_f^X(\rho\mid \sigma)$ is known if the function $f$ is operator convex (resp. operator concave). Some properties have been shown in \cite{P1986}. 
For the special case such that $X:=I$, the following interesting result has been known \cite{HM,M2013}:
\[
f:(0,\infty)\to \mathbb{R}: \text{operator convex implies} \quad S_f^I(A\mid B)\le \hat{S}_f(A\mid B),
\]
where the maximal $f$--divergence was defined by $\hat{S}_f(A\mid B):=\tr Bf\bigl(B^{-1/2}AB^{-1/2}\bigr)$ in \cite{HM,PR}.
Throughout this paper, $I$ represents the identity matrix and $I_n$ represents the $n\times n$ identity matrix, unless otherwise noted. There are other ways to introduce the reduced relative entropy such as
\[
\tr [A\log A-H^*AH\log B]\quad {\rm or}\quad \tr [A\log A-H^*AH\log B- A+ e^{H(\log B)H^*}]
\]
for example. 
In this paper, we study the properties of  \eqref{definition of reduced entropy} for positive definite matrices $A,B$ instead of only density matrices $\rho, \sigma. $  Note that $S_{t\log t}^H(\rho\mid \sigma)\neq S_H(\rho\mid\sigma)$ by \eqref{quasi_form}. 

Earlier the second-named author  \cite{H2015}  gave an interpolation inequality between Golden-Thompson's trace inequality and Jensen's trace inequality
\begin{equation}\label{sec1_eq01}
\tr\left[\exp\left(L+\sum_{j=1}^kH_j^*B_jH_j\right)\right]\le \tr\left[\exp(L)\sum_{j=1}^kH_j^*\exp(B_j)H_j\right],
\end{equation}
valid for self-adjoint $L$, $ B_1,\dots,B_k $ and contractions $H_1,\dots,H_k $ with $ H_1^*H_1+\cdots+H_k^*H_k=I. $

Here we give a variational expression of the  reduced relative  entropy to obtain a lower bound of the  reduced relative entropy. 
A variational expression of the relative  entropy was studied initially in \cite{HP,K1986,P1998}. See also \cite{BPL,BPL2}. 
Using the result in \cite{HP}, we obtain a variational expression for the reduced relative  entropy.

\begin{lemma}\label{theorem01} 
Let $ A,B$ denote $n\times n$ matrices, and let $ H $ be a contraction. For positive definite matrices $X,Y$ we obtain:
\begin{itemize}
\item[(i)] If $A=A^*$, then 
\[
1-\tr\, Y +\log \tr\, e^{A+H\log (Y)H^*}=\max\bigl\{-S_H(X\mid Y)+\tr\, XA\mid \tr\, X=1\bigr\}.
\]
\item[(ii)] If $B=B^*$ and  $\tr\, X=1$, then
\[
S_H(X\mid e^B)=\max\bigl\{\tr\, XA -\log \tr\, e^{A+HBH^*}-1+\tr\, e^{B}\mid  A=A^*\bigr\}.
\]
\end{itemize}
\end{lemma}

\begin{proof} 
In parallel with the reasoning in \cite[Lemma 1.2]{HP} we consider the function
\[
F(X)=\tr\, XA-S_H(X\mid Y)=\tr\, XA-\tr\bigl[X\log X-XH\log(Y)H^*-X+Y\bigr],
\]
defined in positive semi-definite $ X $ with $ \tr\, X=1, $ where we used the cyclicity of the trace and extended the reduced relative  entropy  $ S_H (X\mid Y) $ to positive semi-definite X. This is meaningful since $ Y $ is positive definite and $\lambda\log\lambda\to 0 $ for positive $ \lambda\to 0. $ We also note that $ F $ is a concave function. Since the set of positive semi-definite matrices with unit trace is compact, we obtain that $ F $ attains its maximum in a positive semi-definite matrix $ X_0 $ with unit trace. 
Assume first that $X_0$ is not positive definite. Choose a unit vector $u$ in the kernel of $X_0$ and let $Q$ be the rank one projection onto $u$. For $0<t<1$, set
$X_t:=(1-t)X_0+tQ$ and  compute
\[
\frac{d}{dt}F(X_t)=\tr\, \left[(Q-X_0)\left(A-\log X_t+H\log(Y)H^*\right)\right].
\]
Since $\tr\, Q\log X_t\to -\infty$ and $\tr\, X_0\log X_t\to \tr\, X_0\log X_0$ as $t \searrow 0$
we obtain
 \[
 \lim_{t\searrow 0}\frac{d}{dt}F(X_t)=\infty.
 \] 
 Thus $F(X_t)$ is increasing near $t=0$. This contradicts that $X_0$ is a maximizer of $F,$ so $ X_0 $ is positive definite.
 
Therefore, for any Hermitian $ Z $ with $ \tr\, Z=0, $ we obtain 
\[
\df{}{t} F(X_0+tZ)\left|_{t=0}= \tr\, Z(A-\log X_0+H\log (Y)H^*)\right.=0.
\]
It follows that $ A-\log X_0+H\log (Y)H^* $ is a multiple of the identity. Thus
\[
X_0=\frac{\exp(A+H\log (Y)H^*)}{\tr \exp(A+H\log (Y)H)}
\]
since $ X_0 $ is of unit trace. We write $ X_0=c\cdot \exp(A+H\log (Y)H^*) $ and note that
\[
1=\tr\,X_0=c\cdot\tr\exp\left(A+H\log(Y)H^*\right)
\]
and consequently by taking the logarithm we obtain
\[
\log\tr\exp\left(A+H\log(Y)H^*\right)=-\log c. 
\]
We then insert in
\[
\begin{array}{rl}
F(X_0)&=c \cdot\tr\, A\exp(A+H\log (Y)H^*)\\[2ex]
&-\tr\bigl[c\cdot\exp(A+H\log (Y)H^*) (A+H\log(Y)H^*+\log c)\\[2ex]
&-c\cdot \exp(A+H\log (Y)H^*) H\log(Y)H^*\bigr]+1-\tr\,Y \\[2ex]
&=-c\log c\cdot\tr\exp(A+H\log (Y)H^*)+1-\tr\, Y=-\log c+1-\tr\,Y\\[2ex]
&=\log\tr \exp(A+H\log(Y)H^*) + 1-\tr\, Y,

\end{array}
\] 
and this proves the assertion (i). From (i), the function
\[
g(A)=\log \tr[e^{A+HBH^*}]+1-\tr\,e^B
\]
is convex in Hermitian matrices. 
Thus the function
\[
G(A)=\tr[XA]-\log \tr [e^{A+HBH^*}]-1+\tr e^{B}
\]
is concave in Hermitian matrices.
Let $A_0:=\log X-HBH^*$. Then for every Hermitian $S$ we have
\[
\frac{d}{dt}G(A_{0}+tS)\vert_{t=0}=\tr SX -\frac{\tr e^{A_0+HBH^*}S}{\tr e^{A_0+HBH^*}}=\tr SX-\frac{\tr XS}{\tr X}=0.
\]
From the concavity of $G(A)$ on the Hermitian matrices, this implies that $G(A)$ attains 
the maximum at $A_0$.
Then we obtain
\begin{align*}
& G(A_{0}) =\tr[X(\log X-HBH^*)]-\log \tr[e^{\log X-HBH^*+HBH^*}]-1+\tr e^B\\[0.5ex]
&\qquad \qquad=\tr[X\log X-XHBH^*]-1+\tr\,e^{B}=S_H(X\mid e^B)
\end{align*}
as desired.
\end{proof}

Applying Lemma \ref{theorem01}, we obtain a lower bound of the reduced relative entropy.

\begin{theorem}\label{lower bound of reduced relative entropy}  
Let $X,Y>0$ with $\tr\,X=1$, and let $H$ be a contraction. Then
\[
S_H(X\mid Y)\ge \frac{1}{p}\tr\,\left[H^*XH\log\bigl( Y^{-p/2}X^pY^{-p/2})\right]-\tr [X-Y]-\log \bigl(1+\tr\,[I-HH^*]\bigr)
\]
for $ p>0. $
\end{theorem}

\begin{proof}
 If we set $k:=2$, $ L:=0,\, B_1:=B,\, B_2:=0$, $H_1:=H^*$ and $H_2:=(I-HH^*)^{1/2}$ in  inequality \eqref{sec1_eq01} we obtain
\begin{equation}\label{ineq_trace_Jenasen}
\tr\,[\exp(HBH^*)]\le \tr\,[H\exp (B)H^*]+\tr\,[I-HH^*]
\end{equation}
for a Hermitian $B$ and a contraction $H$. 
By inserting $A=\dfrac{1}{p}H\log \bigl(Y^{-p/2}X^pY^{-p/2}\bigr)H^*$ and $B=\log Y$  in Lemma \ref{theorem01} (ii),
we obtain
\[
\begin{array}{l}
S_H(X\mid Y)+\tr[X-Y]\\[2ex]
\ge \frac{1}{p}\tr\,H^*XH\log (Y^{-p/2}X^pY^{-p/2})-\log \tr\left[\exp\,\left( H\left(\log \bigl(Y^{-p/2}X^pY^{-p/2}\bigr)^{1/p} +\log  Y\right)H^*\right)\right]\\[2ex]
\ge \frac{1}{p}\tr\,H^*XH\log (Y^{-p/2}X^pY^{-p/2})\\[2ex]
-\log \left(\tr\left[H\exp\,\left(\log \bigl(Y^{-p/2}X^pY^{-p/2}\bigr)^{1/p} +\log  Y\right)H^*\right]+\tr\,[I-HH^*]\right)\\[2ex]
\ge \frac{1}{p}\tr\,H^*XH\log (Y^{-p/2}X^pY^{-p/2})\\[2ex]
-\log \left(\tr\left[\exp\left(\log \bigl(Y^{-p/2}X^pY^{-p/2}\bigr)^{1/p} +\log  Y\right)\right]+\tr\,[I-HH^*]\right)\\[2ex]
\ge \frac{1}{p}\tr\,H^*XH\log (Y^{-p/2}X^pY^{-p/2})\\[2ex]
-\log \left(\tr\left[\left(Y^{p/2}\left(Y^{-p/2}X^pY^{-p/2}\right)Y^{p/2}\right)^{1/p}\right]+\tr\,[I-HH^*]\right)\\[2ex]
=\frac{1}{p}\tr\,H^*XH\log (Y^{-p/2}X^pY^{-p/2})-\log \left(1+\tr\,[I-HH^*]\right).
\end{array}
\]
In the second inequality we used inequality \eqref{ineq_trace_Jenasen}.
In the third inequality we used $H^*H\le I$ with $\exp\,\bigl(\log \bigl(Y^{-p/2}X^pY^{-p/2}\bigr)^{1/p} +\log  Y\bigr)\ge 0$, and in the last inequality we used the inequality in \cite[Theorem 1.1]{HP} stating that
\[
\tr\, e^{S+T} \le \tr \bigl(e^{pT/2}e^{pS}e^{pT/2}\bigr)^{1/p},\quad p>0
\]
for Hermitian matrices $S$ and $T$.
\end{proof}

The inequality given in Theorem~\ref{lower bound of reduced relative entropy} may be written as
\begin{equation}
\frac{1}{p}\tr\, H^*XH\log(Y^{-p/2}X^pY^{-p/2}) -\log \left(1+\tr\,[I-HH^*]\right) \le \tr[X\log X-H^*XH\log Y]
\end{equation}
which recovers the inequality in \cite[Theorem 1.3]{HP} by taking $H=I$ and replacing $Y^{-1}$ with $Y$.
If we in addition take $p=1$ in Theorem~\ref{lower bound of reduced relative entropy}, we obtain a lower bound of the reduced relative entropy
\begin{equation}
S_H(X\mid Y)\ge\tr \left[H^*XH\log \left(Y^{-1/2}XY^{-1/2}\right)\right]-\tr [X-Y]-\log \left(1+\tr\,[I-HH^*]\right).
\end{equation}

In Section \ref{sec3}, we define the reduced Tsallis relative  entropy and give a one-parameter extended variational expression. 
In addition, we obtain convexity for the reduced relative entropy. This result gives a simplified proof of the convexity  for the  reduced relative entropy shown in \cite[Theorem 3.1]{H2022}.
Finally we give an upper bound of the  reduced Tsallis relative entropy. This generalizes the existing result.

\section{Main results}\label{sec3} 

The deformed logarithmic function or the $ q $-logarithm is defined by setting
\[
\log_q x=\dfrac{x^{q-1}-1}{q-1}
\]
for $x>0$ and $q\neq 1$. The deformed exponential function $ \exp_q $ is defined as the inverse function of the deformed logarithmic function $ \log_q x. $ It is always positive and given by
\[
\exp_q(x)=
\left\{\begin{array}{ll}
(x(q-1)+1)^{1/(q-1)}\quad&\text{for $ q>1 $ and $x>-(q-1)^{-1} $}\\[1.5ex]
(x(q-1)+1)^{1/(q-1)}\quad&\text{for $ q<1 $ and $x<-(q-1)^{-1} $}\\[1.5ex]
\exp x \quad&\text{for $ q=1$ and  $x\in\mathbb R. $ }
\end{array}\right.
\]
The $ q $-logarithm and the $ q $-exponential functions converge, respectively, to the logarithmic and the exponential functions for $ q\to 1$. We note that
\begin{equation}\label{derivative_q-exponential}
\frac{d}{dx}\log_q(x)=x^{q-2}\qquad\text{and}\qquad \frac{d}{dx}\exp_q(x)=\exp_q(x)^{2-q}\,.
\end{equation}

\begin{definition}\label{the reduced Tsallis relative entropy with parameter }  
For positive definite matrices $A$ and $B$, we define the reduced Tsallis relative entropy with parameter $ q $ by setting
\[
S_{H,\,q}(A\mid B)=\tr\left[A^{2-q}\log_q (A)-H^*A^{2-q}H \log_q (B) -A+B\right]
\]
for $ q \neq 1. $
\end{definition}
It is plain that
 $\lim\limits_{q\to 1}S_{H,q}(A\mid B)=S_{H}(A\mid B).$ 
The reduced Tsallis relative entropy is not non-negative in general, while we do have 
\[
S_{H,q}(UAU^*\mid UBU^*)=S_{U^*HU,q}(A\mid B)
\]
 for a unitary matrix $U$. We intend to study the convexity (concavity) 
 of  the reduced Tsallis relative entropy with applications.
 
We first review Lieb's concavity theorem \cite{L1973} and Ando's convexity theorem \cite{A1979}.
Lieb's concavity theorem states that if $\alpha, \beta \ge 0$ with $\alpha+\beta\le 1$, then for any matrix $X$ the function
$(A,B)\longrightarrow \tr X^*A^{\alpha}XB^{\beta}$ is jointly concave in tuples $ (A,B) $ of positive definite $ A $ and $ B$. Ando's convexity theorem states that necessary and sufficient conditions for joint convexity of $(A,B)\longrightarrow \tr X^*A^{\alpha}XB^{\beta}$  in tuples $ (A,B) $ of positive definite $ A $ and $ B$ are given by $(i),$ $ (ii) $ or $ (iii), $ where
\begin{itemize}
\item[(i)] $-1\le \alpha,\beta\le 0$,
\item[(ii)] $-1\le \alpha \le 0\quad {\rm and}\quad 1-\alpha \le \beta \le 2$,
\item[(iii)] $-1\le \beta \le 0 \quad{\rm and} \quad 1-\beta \le \alpha \le 2$.
\end{itemize}

\begin{theorem}\label{theorem03} 
Take $q\in[0,3]$ with $ q\ne 1, $ and let $H$ be a contraction. If $0\le q <1$ or $1<q\le 2$, then the function
\[
(A,B)\to S_{H,\,q}(A\mid B)
\]
is convex in tuples $ (A,B) $ of positive definite $ A $ and $ B$.  If $2\le q \le 3$, then the function above is concave in tuples $ (A,B) $ of positive definite $ A $ and $ B$. 
\end{theorem}

\begin{proof}  
The reduced Tsallis relative entropy can be written as
\begin{equation}\label{alt_exp_rTRE}
S_{H,q}(A\mid B) =\tr\left[\left(\frac{2-q}{q-1}\right)A+B\right]+\frac{1}{q-1}\tr\left[(HH^*-I)A^{2-q}\right]+\frac{1}{1-q}\tr \left[H^*A^{2-q}HB^{q-1}\right].
\end{equation}
The first term is linear. 
Suppose first that $ q\in [0,1) $ and thus $ 1<2-q\le 2. $
Since $A\to A^{2-q}$ is convex, $ -1\le q-1<0, $ and $HH^*-I\le 0$, the second term is convex.
The third term is also convex since  $\alpha=2-q$ and $\beta=q-1$ satisfy the condition (iii) above in Ando's convexity theorem.

 Suppose now that $q\in(1,2]$. Then we have $ 0\le 2-q< 1 $ and $ 0<q-1\le 1. $ The second term is then convex since $ A\to A^{2-q} $ is concave, and the third term is also convex by Lieb's concavity theorem since $\alpha=2-q$ and $\beta=q-1$ satisfy the condition $\alpha,\beta \ge 0$ and $\alpha+\beta \le 1$.
 
 Finally, we suppose $q\in [2,3]$. The second term is concave since $A\to A^{2-q}$ is matrix convex when $-1\le 2-q\le 0$. The third term is also concave since $\alpha=2-q$ and $\beta=q-1$ satisfy the condition (ii) above in Ando's convexity theorem.
\end{proof}

The following result was proved in \cite[Corollary 2.6]{SH2020}. 

\begin{proposition}\label{proposition_previous_theorem3.2} 

Take $q\neq 1$ and assume $H$ is a contraction. The trace function
\begin{equation}\label{the trace function phi}
\varphi_q(A)=\tr\exp_q\left(L+H^*\log_q (A)H\right)
\end{equation}
is concave in positive definite matrices if $ L\ge 0 $ and $ 1<q\le 2. $  It is convex in positive definite matrices if $ L\ge 0 $ and $2\le q\le 3.$  
\end{proposition}

By refining the arguments in the reference we may strengthing the result to the statement. 

\begin{proposition}  
Let $ H $ be a contraction and take an Hermitian matrix $ L $ such that 
\[
I-H^*H+(q-1)L\ge 0.
\]
 If $1<q\le 2$, then the trace function $ \varphi_q(A) $ defined in (\ref{the trace function phi}) is concave in positive definite matrices. If $ 2\le q\le 3, $ then it is convex in positive definite matrices.
\end{proposition}

\begin{proof}
It follows from an easy calculation that
\[
\varphi_q(A)=\tr\bigl[I-H^*H+(q-1)L+H^*A^{q-1}H\bigr]^{1/(q-1)},
\]
cf. \cite[eq. (3.2)]{H2024}. The remainder of the proof follows as in \cite[Corollary 2.6]{SH2020}. 
\end{proof}

\begin{corollary}\label{cor_LS_concave}  
Let $H$ be a contraction, $L$ and $A$ be positive definite.
 If $q\in(1,2]$, then the function
$
h_{q}(A):=\log_q\varphi_q(A)
$
is concave in positive definite matrices, where $\varphi_q(A)$ is defined in Proposition \ref{proposition_previous_theorem3.2}. If $q\in[2,3]$, then the function
$h_{q}(A)$
is convex in positive definite matrices.
\end{corollary}

\begin{proof}  
The first and second derivatives of the function $\log_{q}(x)$ are given by:
\[
\frac{d}{dx}\log_{q}(x)=x^{q-2} \qquad \frac{d^2}{dx^2}\log_{q}(x)=\left(q-2\right)x^{q-3},\quad (x>0)
\]
from which the statement follows by standard arguments.
\end{proof}

\begin{remark}  
By letting $q \to 1$ from above in Proposition \ref{proposition_previous_theorem3.2}, we note that the assumption
 \[
 I-H^*H+(q-1)L\ge 0
 \]
 is automatically satisfied, and we recover concavity of $A\to \tr \exp\left(L+H^*\log (A)H\right)$ for any Hermitian  matrix $L; $
 a result shown in \cite{H2001}. 
 If we take $H:=I$, then the result \cite[Theorem 6]{L1973} is recovered.
 \end{remark}

We may proceed with Proposition \ref{proposition_previous_theorem3.2} to consider the multivariate extension with block matrices in a similar way as in the paper \cite{H2015}.

\begin{corollary}\label{corollary02} 
Let $H_1,\dots,H_k $ be $ m\times n $ matrices such that $H_1^*H_1+\cdots +H_k^*H_k\le I_n $ and
let $L$ be an $n\times n$ positive semi-definite matrix.  If $q\in(1,2]$, then the function
\[
\varphi(A_1,\cdots,A_k)=\tr\bigl[\exp_q(L+H_1^*\log_q (A_1)H_1+\cdots+H_k^* \log_q (A_k)H_k\bigr]
\]
is concave in positive definite $ m\times m $ matrices $A_1,\dots,A_k\,. $
If $q\in [2,3]$, then $\varphi(A_1,\cdots,A_k)$ is convex  in positive definite $ m\times m $ matrices $A_1,\dots,A_k\,. $
\end{corollary}

\begin{proof} 
Consider the $k\times k$ block matrices:
\[
\hat A=\begin{pmatrix}
     A_1     & 0     & \cdots   & 0\\
     0         & A_2 &             & 0\\
     \vdots &         & \ddots  & \vdots\\
     0        & 0       & \dots    & A_k
     \end{pmatrix},\quad
     \hat L=\begin{pmatrix}
     L     & 0     & \cdots   & 0\\
     0         & 0  &             & 0\\
     \vdots &         & \ddots  & \vdots\\
     0        & 0       & \dots    & 0
     \end{pmatrix},\quad
\hat H=\begin{pmatrix}
     H_1    & 0        & \cdots & 0\\
     H_2    & 0        & \cdots & 0\\
     \vdots & \vdots & \ddots & \vdots\\
     H_k    & 0         & \cdots & 0
     \end{pmatrix}
\]
and calculate
\[
\hat{L} + {\hat{H}^*}{\log _q}\left( \hat{A} \right)\hat{H} =
\begin{pmatrix}
     L + \sum\limits_{i = 1}^k {H_i^*{{\log }_q}\left( {{A_i}} \right){H_i}}     & 0     & \cdots   & 0\\
     0         & 0 &             & 0\\
     \vdots &         & \ddots  & \vdots\\
     0        & 0       & \dots    & 0
     \end{pmatrix},\quad
\]
from which we obtain the identity
\[
\tr \exp_q\left(\hat{L} + {\hat{H}^*}{\log _q}\left( \hat{A} \right)\hat{H}\right)  
=\tr \exp_q\left(  L + \sum\limits_{i = 1}^k {H_i^*{{\log }_q}\left( {{A_i}} \right){H_i}} \right) +(k-1)n
\]
which is concave (resp. convex) for $q\in(1,2]$ (resp. $q\in[2,3]$) by Proposition \ref{proposition_previous_theorem3.2} with $L\ge 0$. This proves the statements in the corollary. 
\end{proof}
Take $q>1$ and $ L, A_1,\cdots,A_k >0$. 
By the definition of the deformed exponential we obtain
\[ 
\varphi(A_1,\cdots,A_k) =\tr \left[I-\sum_{i=1}^kH_i^*H_i+(q-1)L+\sum_{i=1}^kH_i^*A_i^{q-1}H_i\right]^{\frac{1}{q-1}}.
\]
We already considered the case $q\in(1,3]$ in Proposition \ref{proposition_previous_theorem3.2} and Corollaries \ref{cor_LS_concave}, \ref{corollary02}. In the remainder of the paper, we consider the cases $q\in(1,2]$ and $q\in[0,2]\backslash\{1\}.$ 
We introduce the following Jensen type inequality.
\begin{proposition}\label{theorem05} 
Let $H_1,\dots,H_k $ be $ m\times n $ matrices such that $H_1^*H_1+\cdots +H_k^*H_k= I_{n} $  and let $B_j,\,\,(j=1,2,\cdots,k)$ be  $m\times m$ positive definite matrices.  We then obtain the inequality
\begin{equation}\label{ineq_theorem3.2}
\tr\left[\exp_{q}\left(\sum_{j=1}^kH_j^*B_jH_j\right)\right]\le \tr\left[\sum_{j=1}^kH_j^*\exp_{q}(B_j)H_j\right]
\end{equation}
for $ 1\le q\le 2. $
\end{proposition}

\begin{proof} We first notice that all the relevant matrices are in the domain of $ \exp_q. $
The function $ \exp_{q}(x) $ is convex in $ x\ge 0 $ for $ 1\le q\le 2, $ since $ \exp_q(x)>0 $ and
\[
\dfrac{d^2}{dx^2}\exp_{q}(x)=\dfrac{d}{dx} \exp_q(x)^{2-q}=(2-q)\exp_q(x)^{1-q}\exp_q(x)^{2-q}\ge 0.
\]
The inequality then follows from Jensen's trace inequality \cite[Theorem 2.4]{H2003}.
\end{proof}

A variational expression for  the reduced Tsallis relative entropy can be obtained with a similar reasoning as in Lemma \ref{theorem01} with some complicated calculations. To this end, we prepare the following lemma.

\begin{lemma}\label{lemma_commute_derivative} 

Consider a  $t\in\mathbb{R}.$ 

\begin{itemize}

\item[(i)]  Let $f:\mathbb{R}\to\mathbb{R}$ be a continuously differentiable function.
For any $n \times n$ Hermitian matrix $A$ and $B$ the derivative
\[
\left.\df{}{t}\tr f(A+tB) \right|_{t=0}=\tr f'(A)B.
\]
\item[(ii)] Let $X,Y$ and $Z$ be $n\times n$ Hermitian matrices such that $XZ=ZX$, and let
$f$ be a continuously differentiable function defined in an open interval containing the eigenvalues of $X$.
Then,
\[
\left. \frac{d}{dt}\tr f\left(X+tY\right)Z \right|_{t=0}=\tr\, Y f'(X)Z.
\]
\end{itemize}
\end{lemma}

\begin{proof} 
The fact (i) follows by taking the trace on both sides in  \cite[Theorem 3.2]{H1995}, or by \cite[Theorem 3.23]{HP2014}.

Next, we prove (ii). We can take simultaneous diagonalizations $X=U{\rm diag(\lambda_1,\cdots,\lambda_n)}U^*$ and 
$Y=U{\rm diag(\mu_1,\cdots,\mu_n)}U^*$ with a unitary matrix $U$. Thanks to the Daleckii--Krein derivative formula (see for example \cite{bhatia2007}) we obtain
\[
\left. \frac{d}{dt}\tr f\left(X+tY\right) \right|_{t=0}=U\left([f^{[1]}(\lambda_i,\lambda_j)]_{i,j=1}^n \circ (U^*YU)\right)U^*,
\]
where $\circ$ denotes the Schur product and
\[
f^{[1]}(x,y): = \left\{\begin{array}{ll}
\displaystyle\frac{f(x)- f(y)}{x - y}\qquad &x \ne y \\[2.5ex]
f'(x) &x = y. 
\end{array} \right.
\]
We therefore obtain
\begin{align*}
\left.\frac{d}{dt}\tr f(X+tY)Z\right|_{t=0}&=\tr \,\,U\left([f^{[1]}(\lambda_i,\lambda_j)]_{i,j=1}^n \circ (U^*YU)\right)U^*U{\rm diag(\lambda_1,\cdots,\lambda_n)}U^*\\[0.5ex]
&=\tr\,\, {\rm diag(\lambda_1,\cdots,\lambda_n)}\left([f^{[1]}(\lambda_i,\lambda_j)]_{i,j=1}^n \circ (U^*YU)\right)\\[0.5ex]
&=\tr \left[\mu_i f^{[1]}(\lambda_i,\lambda_j)(U^*YU)_{ij}\right]_{i,j=1}^n\\[0.5ex]
&=\sum_{i=1}^n\mu_if'(\lambda_i)(U^*YU)_{ii}\\[0.5ex]
&=\tr\,\, {\rm diag}\left(\mu_1f'(\lambda_1),\cdots, \mu_nf'(\lambda_n)\right)U^*YU\\[0.5ex]
&=U^*f'(X)ZUU^*YU=\tr\,\, Yf'(X)Z
\end{align*}
as desired.
\end{proof}

\begin{theorem}\label{theorem06} 

Let $A,B,X,Y$ be $n\times n$ matrices, and let $H$ be a contraction; take $q\in(1,2]$ and $\gamma>0.$
\begin{itemize}
\item[(i)] If  $Y>0$ and $A=A^*$ with $A+H\log_q(Y)H^*>-\dfrac{1}{q-1}I$, then 
\[
\begin{array}{l}
\displaystyle \gamma \log_{q}\left[\gamma^{-1}\tr\exp_{q}\bigl(A+H \log_{q} (Y)H^*\bigr)\right]+\gamma -\tr\,[Y]\\[2.5ex]
=\max\left\{\tr[X^{2-q}A]-S_{H,\,q}(X\mid Y)\right\},
\end{array}
\]
where the maximum is taken over positive definite $X$ with $\tr\, X=\gamma$.

\item[(ii)] If $X \ge 0$ with $\tr X=\gamma$ and $B=B^*$ with $\log_q X> HBH^*$ and $B>-\dfrac{1}{q-1}I$, then 
\[
\begin{array}{l}
S_{H,q}\left(X\mid \exp_{q} B\right)\\[2ex]
=\displaystyle\max\left\{\tr[X^{2-q}A]-\gamma \log_{q}\left[\gamma^{-1}\tr\exp_{q}\left(A+HBH^*\right)\right] -\gamma+\tr\exp_{q}B\right\},
\end{array}
\]
where the maximum is taken over  positive definite $A$.

\end{itemize}
\end{theorem}

\begin{proof} 

To prove (i), we define 
\[
F_{q}(X)=\tr[X^{2-q}A]-S_{H,\,q}(X\mid Y)
\]
for positive definite matrices $ X $ with $\tr\,X=\gamma$. 
The reduced Tsallis relative entropy is written by the form \eqref{alt_exp_rTRE}. Assume $X_k\to X$ for $X_k\ge 0,\,\,(k\in\mathbb{N})$. Then $X_k^{2-q}\to X^{2-q}$ for $q\in(1,2]$. 
Thus $F_{q}(X)$ is continuous in the compact set ${\frak D}_{n,\gamma}$ of all $n \times n$ positive semidefinite matrices with  trace $\gamma$. Therefore, $F_{q}(X)$ takes maximum in a certain $X_0$ with $\tr\, X_0=\gamma$ in ${\frak D}_{n,\gamma}\,.$

We note from the assumptions that
\[
I+(q-1)\left(A+H(\log_q Y)H^*\right)>0
\]
 for $1<q\le 2$,\,\, $A=A^*$ and $Y>0, $ so 
$\exp_{q}\left(A+H\log_{q}(Y)H^*\right)$ is well defined.
Let
\[
X_0=\frac{\gamma\,\, \exp_{q}\left(A+H\log_{q}(Y)H^*\right)}{\tr\left[\exp_{q}\left(A+H\log_{q}(Y)H^*\right)\right]}\,.
\]
For any Hermitian matrix $S$ with $\tr\,S=0$, we have 
\begin{align*}
F_{q}(X_0+tS)=\tr\left[(X_0+tS)^{2-q}\left(A+\frac{1}{q-1}I+H\log_q(Y)H^*\right)\right]-\frac{2-q}{q-1}\tr [X_0+tS]-\tr\, Y.
\end{align*}
Since $X_0$ and $A+\dfrac{1}{q-1}I+H\log_q (Y)H^*$ commute, it follows from Lemma \ref{lemma_commute_derivative} (ii) that
\[
\left.\frac{d}{dt}F_{q}(X_0+tS)  \right|_{t=0}
=\left(2-q\right)\tr\,\, SX_0^{1-q}\left(A+\frac{1}{q-1}I+H\log_{q}(Y)H^*\right)-\frac{2-q}{q-1}\tr\,S.
\]
By setting $c:=\tr\,\,\exp_q\left(A+H\log_q(Y)H^*\right)$, we obtain
\[
\frac{c}{\gamma}X_0=\exp_q\left(A+H\log_q(Y)H^*\right),
\]
that is,
\[
\frac{\left(\frac{c}{\gamma}X_0\right)^{q-1}}{q-1}=A+\frac{1}{q-1}I+H\log_q(T)H^*,
\]
so that
\[
\frac{\left(\frac{c}{\gamma}\right)^{q-1}}{q-1}=X_0^{1-q}\left(A+\frac{1}{q-1}I+H\log_q(T)H^*\right).
\]
Therefore, we have
\[
\left.\frac{d}{dt}F_{q}(X_0+tS)  \right|_{t=0}=0
\]
as $\tr\,S=0$. Since furthermore $F_q(X)$ is concave in the set of positive definite matrices, the function attains maximum in $X_0, $ and since
\begin{equation}\label{formula_logq}
\log_{q}\left(\dfrac{y}{x}\right)=\log_{q} y-\left(\dfrac{y}{x}\right)^{q-1}\log_{q} x
\end{equation}
and by setting $K_{q}=\exp_{q}\left(A+H\log_{q}(Y)H^*\right), $ we obtain
\begin{align*}
&F_{q}(X_0)=\tr[X_0^{2-q}A]-\tr[X_0^{2-q}\log_{q} X_0]+\tr[X_0^{2-q}H\log_{q} (Y)H^*]+\tr\,X_0-\tr\,Y\\[2ex]
&=\gamma^{2-q}\frac{\tr\left[K_{q}^{2-q}\left\{A-\log_{q}\left(\frac{\gamma \,\,K_{q}}{\tr\,K_{q}}\right)+H\log_{q} (Y)H^*\right\}\right]}{(\tr\,K_{q})^{2-q}}+\gamma-\tr\,Y\\[2ex]
&= \gamma^{2-q}\frac{\tr\left[K_{q}^{2-q} \left(\frac{\gamma\,\, K_{q}}{\tr\,K_{q}}\right)^{{q-1}}\log_{q} \left( \frac{1}{\gamma}\,\,\tr\,K_q\right)\right]}{(\tr\,K_{q})^{2-q}}+\gamma-\tr\, Y\quad \text{(by \eqref{formula_logq})}\\[2ex]
&=\gamma \,\log_{q}\left[\gamma^{-1}\,\, \tr\exp_{q}\left(A+H\log_{q} (Y)H^*\right)\right] +\gamma-\tr\,Y.
\end{align*}

Next, we prove (ii).  We note that 
\[
A+HBH^*>HBH^*\ge -\frac{1}{q-1}HH^*\ge -\frac{1}{q-1}I
\]
so that $\exp_q\left(A+HBH^*\right)$ is well defined for $1<q\le 2$.
It  follows by (i) and by using the triangle inequality of $\max$ that the functional 
\[
g_{q}(A)=\gamma \,\log_{q}\left[\gamma^{-1}\,\,\tr\exp_{q}\left(A+HBH^*\right)\right] +\gamma-\tr\exp_{q}B
\]
is convex. 
Then functional $G_{q}(A):=\tr\, X^{2-q}A-g_{q}(A)$ is then concave in the set of all positive definite matrices. 
Let $A_0:=\log_q X-HBH^*$ ($>0$ by assumption). For any Hermitian matrix $S$,
we obtain by Lemma \ref{lemma_commute_derivative} (i), see also equation \eqref{derivative_q-exponential}, that
\begin{align*}
&\left. \frac{d}{dt}G_{q}(A_0+tS)\right|_{t=0} \\
&\quad =\tr\,X^{2-q}S-\gamma\left(\gamma^{-1}\tr\,\exp_q\left(A_0+HBH^*\right)\right)^{q-2}\cdot
\gamma^{-1}\tr\,\left(\exp_q\left(A_0+HBH^*\right)\right)^{2-q}S\\
&\quad =\tr\,X^{2-q}S-\left(\gamma^{-1}\tr \, X\right)^{q-2}\tr\,X^{2-q}S\\
&\quad =\tr \,X^{2-q}S-\tr\,X^{2-q}S=0.
\end{align*}
Thus $G_{q}$ takes maximum in $A_0$ and we obtain
\begin{align*}
& G_{q}(A_0)=\tr[X^{2-q}\left(\log_{q} X-HBH^*\right)]-\gamma\,\log_{q}\left(\gamma^{-1}\,\,\tr X\right)-\gamma+\tr\exp_{q}(B)\\
&=\tr[X^{2-q}\left(\log_{q} X-H\log_{q} (\exp_{q} B)H^*\right)]-\tr\, X +\tr \exp_{q}(B)\\
&= S_{H,\,{q}}\left(X\mid \exp_{q} B\right)
\end{align*}
as minimal value.
\end{proof}

To prove the next proposition, we first recall the parametric extended Golden-Thompson inequality \cite[Proposition 3.2]{FL2010}. 
\begin{lemma}\label{lemma01} 

For $q\in(1,2]$ the inequality
\[
\tr[\exp_q(A+B)] \le \tr[\exp_q(A)\exp_q(B)]
\]
is valid for positive semi-definite matrices. 
\end{lemma}

\begin{proposition}\label{proposition02} 
Assume $I\le Y \le X$ with $\tr\, X=:\gamma$ and $q\in(1,2]$. Then 
\[
S_{H,\,q}(X\mid Y)\ge \tr\left[H^*X^{2-q}H\log_{q}\bigl(Y^{-1/2}XY^{-1/2}\bigr)\right]-\tr[X-Y]-\gamma\,\log_{q}\bigl(1+\gamma^{-1}\,\,\tr\,[I-HH^*]\bigr).
\]
\end{proposition}

\begin{proof} 

We take $k=2$, $B_1=B,\, B_2=0$, $H_1=H^*$ and $H_2=(I-HH^*)^{1/2}$ in inequality \eqref{ineq_theorem3.2} and obtain
 \begin{equation}\label{ineq_trace_Jenasen_q}
\frac{1}{\gamma}\,\tr\,[\exp_{q}\left(HBH^*\right)]\le \frac{1}{\gamma}\,\tr\,[H\exp_{q}(B)H^*]+\frac{1}{\gamma}\,\tr\,[I-HH^*],\quad 1<q\le 2
\end{equation}
for positive definite $B$ and contractions $H. $
Insert $B=\log_{q} Y$ and $A=H\log_{q}\bigl(Y^{-1/2}XY^{-1/2}\bigr)H^*$ in Theorem \ref{theorem06} (ii). The conditions $Y\ge I$ and $X \ge Y$ assure that $B \ge 0$ and $A \ge 0.$ By the same reasoning as in the proof of Theorem \ref{lower bound of reduced relative entropy} with inequality \eqref{ineq_trace_Jenasen_q} and Lemma \ref{lemma01},  we obtain
\[
\begin{array}{l}
S_{H,\,q}(X\mid Y)+\tr[X-Y] \\[1.5ex]
 \ge\tr\bigl[H^*X^{2-q}H\log_{q}\bigl(Y^{-1/2}XY^{-1/2}\bigr)\bigr]\\[2ex]
 -\gamma\,\log_{q}\bigl(\gamma^{-1}\,\,\tr\bigl[\exp_{q}\bigl(H(\log_{q}(Y^{-1/2}XY^{-1/2})+\log_qY)H^*\bigr)\bigr]\bigr)\\[2ex]
\ge \tr\left[H^*X^{2-q}H\log_{q}\bigl(Y^{-1/2}XY^{-1/2}\bigr)\right]\\[2ex]
-\displaystyle \gamma\,\log_{q} \left(\gamma^{-1}\,\,\tr\left[H\exp_{q}\left(\log_q(Y^{-1/2}XY^{-1/2})+\log_{q}Y\right)H^*\right]+\gamma^{-1}\,\,\tr\,[I-HH^*]\right)\\[2ex]
\ge \tr\left[H^*X^{2-q}H\log_{q}\bigl(Y^{-1/2}XY^{-1/2}\bigr)\right]\\[2ex]
-\displaystyle \gamma\,\log_{q} \left(\gamma^{-1}\,\,\tr\left[\exp_{q}\left(\log_{q}(Y^{-1/2}XY^{-1/2})+\log_{q}Y\right)\right]+\gamma^{-1}\,\,\tr\,[I-HH^*]\right)  \\[2ex]
\ge \tr\left[H^*X^{2-q}H\log_{q}\bigl(Y^{-1/2}XY^{-1/2}\bigr)\right]\\[2ex]
 -\displaystyle \gamma\,\log_{q} \left(\gamma^{-1}\,\,\tr\left[\exp_{q}\left(\log_{q} (Y^{-1/2}XY^{-1/2})\right)\exp_{q}\left(\log_{q} Y\right)\right]+\gamma^{-1}\,\,\tr\,[I-HH^*]\right) \\[2ex]
=\tr\left[H^*X^{2-q}H\log_{q}\bigl(Y^{-1/2}XY^{-1/2}\bigr)\right]-\gamma\,\log_{q}\left(1+\gamma^{-1}\,\,\tr\,[I-HH^*]\right),
\end{array}
\]
where we used $H^*H\le I. $ 
\end{proof}

We close this paper by giving an upper bound of the reduced Tsallis relative entropy.
We make use of the following lemma. 

\begin{lemma}\label{lemma_BPL_FS} 

Let $A$ and $B$ be positive definite matrices. Then
\begin{itemize}
\item[(i)] $\tr \left[A^{1+t} B^t \right]\le \tr \left[A\left(A^{s/2}B^sA^{s/2}\right)^{t/s} \right]$ for $s \ge t > 0$.
\item[(ii)] $\tr \left[A\left(A^{-s/2}B^sA^{-s/2}\right)^{t/s} \right] \le \tr \left[A^{1-t}B^t\right]$ for $s\ge t > 0$ and $0<t\le 1$.
\end{itemize}
\end{lemma}
These results were obtained in  \cite[Theorem 2.1]{BPL2} and \cite[Theorem 3.1]{FS2021}. 

\begin{theorem}\label{theorem_upper_bound}  

Let $A$ and $B$ be positive definite matrices, and let $H$ be an invertible contraction. Then
\begin{equation}\label{theorem_upper_bound_eq01}
\begin{aligned}
&S_{H,\,q}(A\mid B)-\frac{1}{q-1}\tr \left[(HH^*-I)A^{2-q}\right]+\tr [A-B]\\
&\le -\tr \left[A\log_{q}\left\{A^{-p/2}\left(HB^{q-1}H^*\right)^{p/(q-1)}A^{-p/2}\right\}^{1/p}\right]
\end{aligned}
\end{equation}
for $q \in [0,2]\backslash{\left\{1\right\}}$ and $p\ge |q-1| >0$.
\end{theorem}

\begin{proof} 

Let $-1\le r <0$. If $p\ge -r >0$, then we have
\begin{align*}
\tr \left[A^{1-r}HB^rH^*\right]&=\tr \left[A^{1-r}\left\{\left(HB^rH^*\right)^{-1/r}\right\}^{-r}\right]\\
&\le \tr \left[A\left\{A^{p/2}\left(HB^rH^*\right)^{-p/r}A^{p/2}\right\}^{-r/p}\right]\quad \text{(by Lemma \ref{lemma_BPL_FS} (i))}\\
&= \tr \left[A\left\{A^{-p/2}\left(HB^rH^*\right)^{p/r}A^{-p/2}\right\}^{r/p}\right].
\end{align*}
Set $r:=q-1$. By use of the above we obtain
\begin{align*}
& S_{H,\,q}(A\mid B)-\frac{1}{q-1}\tr \left[(HH^*-I)A^{2-q}\right]+\tr [A-B] \\
&= -\frac{1}{q-1} \tr \left[A^{2-q}HB^{q-1}H^*-A\right]\\
&\le - \frac{1}{q-1}\tr \left[A\left\{A^{-p/2}\left(HB^{q-1}H^*\right)^{p/(q-1)}A^{-p/2}\right\}^{(q-1)/p}-A\right]\\
&= -\tr \left[A\log_{q}\left\{A^{-p/2}\left(HB^{q-1}H^*\right)^{p/(q-1)}A^{-p/2}\right\}^{1/p}\right]
\end{align*}
for $p\ge 1-q>0$ and $0\le q <1. $ Let next $0<r \le 1$. If $p\ge r >0$ we obtain 
\begin{align*}
\tr \left[A^{1-r}HB^rH^*\right]&= \tr \left[A^{1-r}\left\{\left(HB^rH^*\right)^{1/r}\right\}^{r}\right]\\
&\ge \tr \left[A\left\{A^{-p/2}\left(HB^rH^*\right)^{p/r}A^{-p/2}\right\}^{r/p}\right] \quad \text{(by Lemma \ref{lemma_BPL_FS} (ii))} 
\end{align*}
for $p\ge r >0. $
Next, set $r=q-1. $ By use of the above we obtain
\begin{align*}
& S_{H,\,q}(A\mid B)-\frac{1}{q-1}\tr \left[(HH^*-I)A^{2-q}\right]+\tr [A-B] \\
& = -\frac{1}{q-1} \tr \left[A^{2-q}HB^{q-1}H^*-A\right]\\
& \le -\frac{1}{q-1}\tr \left[A\left\{A^{-p/2}\left(HB^{q-1}H^*\right)^{p/(q-1)}A^{-p/2}\right\}^{(q-1)/p}-A\right] \\
&= -\tr \left[A\log_{q}\left\{A^{-p/2}\left(HB^{q-1}H^*\right)^{p/(q-1)}A^{-p/2}\right\}^{1/p}\right]
\end{align*}
for $p\ge q-1 >0$ and $1< q \le 2. $
\end{proof}


Setting $H=I$ and $\alpha = q-1$ in Theorem \ref{theorem_upper_bound}, inequality \eqref{theorem_upper_bound_eq01} recovers the results established in \cite[Theorem 2.3]{Seo} and \cite[Thereom 4.1]{FS2021} stating that
$$
\tr\left[\frac{A-A^{1-\alpha}B^{\alpha}}{\alpha}\right] \le -\tr \left[A\log_{1+\alpha}\left(A^{-p/2}B^pA^{-p/2}\right)^{1/p}\right]
$$
for $\alpha \in [-1,1]\backslash{\left\{0\right\}}$ and $p\ge |\alpha| >0$.

\subsection*{Declarations}
\begin{itemize}
\item {\bf{Availability of data and materials}}: Not applicable.
\item {\bf{Competing interests}}: The authors declare that they have no competing interests.
\item {\bf{Funding}}: This research is supported by a grant (JSPS KAKENHI, Grant Number: JP21K03341) awarded to the author, S. Furuichi.
\item {\bf{Authors' contributions}}: Authors declare that they have contributed equally to this paper. All authors have read and approved this version.
\end{itemize}


{\tiny (S. Furuichi) Department of Information Science, College of Humanities and Sciences, Nihon University, Setagaya-ku, Tokyo, Japan},

{\tiny Department of Mathematics, 
Saveetha School of Engineering, SIMATS,
Thandalam, Chennai -- 602105,
Tamilnadu, India}

{\tiny \textit{E-mail address:} furuichi.shigeru@nihon-u.ac.jp}

\vskip 0.3 true cm

{\tiny (F. Hansen) Department of Mathematics, University of Copenhagen,	Universitetsparken 5, 2100 Copenhagen, Denmark}

{\tiny \textit{E-mail address:} frank.hansen@math.ku.dk}

\end{document}